\documentclass[journal]{IEEEtran}
\usepackage{palatino,epsfig,fleqn,amssymb,cite,color}
\usepackage{amsmath}
\usepackage{bm,cases,subcaption}

\newtheorem{prop}{Proposition}
\newtheorem{lemma}{Lemma}
\newtheorem{theorem}{Theorem}
\newtheorem{remark}{Remark}

\newtheorem{defn}{Definition}
\newtheorem{ass}{Assumption}
\def\QED{~\rule[-1pt]{5pt}{5pt}\par}
\newenvironment{proof}{{\em Proof.}}{{ \hfill \QED}\medskip}

\DeclareMathAlphabet{\matheur}{U}{eur}{m}{n}
\DeclareMathAlphabet{\matheurb}{U}{eur}{b}{n}
\DeclareMathAlphabet{\matheus}{U}{eus}{m}{n}
\DeclareMathAlphabet{\matheuf}{U}{euf}{m}{n}

\newcommand{\Rset}{\mathbb{R}}

\newcommand{\bfDel}{{\bf\Delta}}

\newcommand{\hs}{\hspace{4mm}}

\renewcommand{\t}{^{\mbox{\tiny\sf T}}}
\newcommand{\IC}{\mathbb{C}}
\newcommand{\IR}{\mathbb{R}}

\newcommand{\IU}{\mathbb{U}}

\newcommand{\II}{\mathbb{I}}

\newcommand{\rank}{{\rm rank}}
\newcommand{\diag}{{\rm diag}}
\newcommand{\col}{{\rm col}}

\newenvironment{mat}{\left[\begin{array}}{\end{array}\right]}

\definecolor{bg}{rgb}{0,0.4,0.4}		%
\definecolor{rg}{rgb}{0.3,0.7,0}		%

\newcommand{\rhoU}{{\bar\rho}}
\newcommand{\rhoL}{{\varrho}}

\DeclareMathAlphabet{\matheur}{U}{eur}{m}{n}
\DeclareMathAlphabet{\matheurb}{U}{eur}{b}{n}
\DeclareMathAlphabet{\matheus}{U}{eus}{m}{n}
\DeclareMathAlphabet{\matheuf}{U}{euf}{m}{n}

\newcommand{\LFT}{\mathcal{F}_\ell}

\newcommand{\RHinf}{\mathcal{RH}_\infty}
\newcommand{\RLinf}{\mathcal{RL}_\infty}     %
\begin{document}

\title
{Robust Instability Radius for Networked Dynamical Systems: Upper and Lower Bounds}

\author{Shinji Hara$^{1}$~\IEEEmembership{Fellow, IEEE}, Yutaka Hori$^{2}$~\IEEEmembership{Senior Member, IEEE}, Tetsuya Iwasaki$^{3}$~\IEEEmembership{Fellow, IEEE}, Chung-Yao Kao$^{4}$~\IEEEmembership{Member, IEEE} and Sei Zhen Khong$^{4}$~\IEEEmembership{Senior Member, IEEE}
\thanks{This work was supported in part by the National Science and Technology Council of Taiwan (grant numbers: 
113-2221-E-110-048-MY3, 113-2918-I-110-003, 113-2222-E-110-002-MY3, 114-2622-8-110-00, 115-2218-E-007-003, and 115-2221-E-110-058-MY2).
The authors are listed alphabetically. Corresponding author: C.-Y. Kao.}%
\thanks{$^{1}$ S. Hara is with the Supercomputing Research Center, Institute of Integrated Research, Institute of Science Tokyo, Tokyo 152-8550, Japan. {\tt shinji\_hara@ipc.i.u-tokyo.ac.jp} }
\thanks{$^{2}$ ~Y.~Hori is with Applied Physics and Physico-Informatics, Keio University, 
3-14-1 Hiyoshi, Kohoku-ku, Yokohama, Kanagawa 223-8522, Japan.
{\tt \small yhori@appi.keio.ac.jp}. } 
\thanks{$^{2}$ T.~Iwasaki is with Mechanical and Aerospace Engineering, University of California Los Angels, 
420 Westwood Plaza, Los Angeles, CA 90095, USA. 
{\tt \small tiwasaki@ucla.edu}, }
\thanks{$^{4}$ Chung-Yao Kao and Sei Zhen Khong are with the Department of Electrical Engineering, National Sun Yat-sen University, 
          Kaohsiung 804201, Taiwan. {\tt \{cykao, szkhong\}@mail.nsysu.edu.tw} }%
}

\maketitle

\begin{abstract}
This paper is concerned with robust instability of uncertain network systems.
We consider the multi-agent system described as a network of single-input-single-output
agents with identical nominal dynamics subject to heterogeneous perturbations.
The network description is formalized as a feedback interconnection of a diagonal 
uncertainty, nominal identical agents, and a static interconnection matrix.
Assuming that the nominal network is unstable, we seek the robust instability
radius (RIR), defined as the smallest norm of the stable uncertainty that renders
the network stable. Conditions for the network stability are developed, and 
upper and lower bounds on the RIR are derived. When the network connectivity matrix
is rank one and all diagonal entries share the same sign or are zero, we give conditions 
under which the RIR is exactly characterized by a small gain argument.
\end{abstract}

\begin{IEEEkeywords}
Robust instability, Stability analysis, Multi-agent network systems
\end{IEEEkeywords}

\section{Introduction} 
\label{sec:Intro}

Periodic oscillation phenomena and regulation dynamics are fundamental to biological systems.  
Well known examples of such phenomena include Repressilator \cite{Elowitz2000} in synthetic biology, 
periodic bursting of spikes in neuronal circuits \cite{ijspeert:08}, %
and periodic pattern generation by Turing instability \cite{YMKH:MBMC2015}. 
Some of those oscillatory behaviors emerge from network dynamics of multiple agents, and 
are sustained by instability of an equilibrium point. 
The agent dynamics are subject to perturbations which may stabilize the equilibrium, 
resulting in the loss of oscillations.
 
Therefore, similar to robust stability analysis, robust instability analysis against unmodeled dynamics 
and/or uncertainties is important to theoretically guarantee preservation of 
desirable oscillation phenomena.  

Following the authors' previous work on scalar systems 
\cite{hara2022instability,kao:23,hara2023exact}, 
this paper is concerned with the robust instability problem for uncertain networked dynamical systems. 
The robust instability problem is similar to but different from the 
standard robust stability analysis for which the small gain theorem plays a main role.
It is mathematically equivalent to a strong stabilization problem \cite{Youla:Automatica1974} 
with a minimum norm controller. The strong stabilization problem itself is difficult to solve, 
despite there being a nice necessary and sufficient condition for the existence of stable stabilizing controller, known as the Parity Interlacing Property (PIP) condition.  
This clearly indicates, with the added minimum norm requirement, 
that the robust instability problem is very difficult to get the exact solution. 
Therefore, the previous work \cite{hara2023exact} mainly focused on the question 
"{\em Under what condition does the small gain argument provide an exact solution}
for the robust instability analysis?" and provided a partial answer to the question by showing that 
the robust instability problem can be reduced to a new optimization problem, called 
"Phase Change Rate (PCR) Maximization Problem". 
 
The network version of robust instability analysis is a counterpart to that for robust stability 
in \cite{hara:19}.  Similar to the scalar system case, it is far more difficult than 
the robust stability analysis to derive the exact solutions. Hence, this paper gives some 
foundational results for \cite{NetworkRIR} that provides exact solutions for a certain class of networked systems. 
To this end, we first formulate a problem for the robust instability analysis, and 
examine internal stability of a class of dynamical systems consisting of 
a nominal networked system and uncertain block. 
One of the main results of this paper is the derivation of
upper and lower bounds on the robust instability radius (RIR). 
Another contribution is a detailed analysis for the case where the network 
graph is rank deficient. We provide a two-step optimization procedure to derive an improved upper bound, which is equal to the RIR for the rank-one case. 

This paper is organized as follows. 
Section~\ref{sec:ProblemFormulation} is devoted to the problem formulation for robust 
instability analysis.  Internal stability of a class of uncertain networked dynamical 
systems is analysed in Section~\ref{sec:bsa}.
Then, upper and lower bounds on the robust instability radius  (RIR) are provided in 
Section~\ref{sec:Preliminaries}. 
Section~\ref{sec:RankDeficient} examines the case of rank-deficient interconnection matrices. 
Section~\ref{sec:Concl} summarizes the paper and   
details some connections with a paper on the exact solution to the robust instability problem~\cite{NetworkRIR}.

\noindent 
{\bf Notation:}  
The set of real numbers is denoted by $\IR$, and 
the set of integers $1,\ldots,n$ is denoted by $\II_n$. 
$\Re(s)$ and $\Im(s)$ denote the real and imaginary parts of a complex number $s$, respectively. 
The open (closed) left and right half complex planes are abbreviated as OLHP (CLHP) and ORHP (CRHP), respectively. Denote by
\begin{itemize}
\item
$\mathcal{R}_p$ : the set of real-rational proper transfer functions 
\item
$\mathcal{RL}_\infty$ : the set of real-rational proper transfer functions with no poles on the imaginary axis 
\item
$\mathcal{RH}_\infty$ : the set of real-rational proper transfer functions with no poles in the CRHP.
\end{itemize}
\noindent
The norm of $\RLinf$ function is denoted by $\| \cdot \|_{\infty}$. The dimensions of transfer function sets are indicated by superscripts, e.g. $\RHinf^{n\times m}$.

\section{Problem Formulation}
\label{sec:ProblemFormulation}
\subsection{Uncertain Network Structure}
\label{subsec:UncertainDNS}

Consider a multi-agent system with $n$ uncertain dynamic SISO agents. 
We assume that all the nominal agents share the same transfer 
function, say $h(s)$, which is a strictly proper real rational function. 
However, each agent is subject to heterogeneous uncertainty represented by $\delta_i(s)$, and the dynamics of the $i$th agent are described by an upper linear fractional 
representation 
\[
{\cal F}_u(\delta_i, W) = h + 
\frac{w_{12}w_{21}\delta_i}{1 - w_{11}\delta_i} , \hs
W = \begin{mat}{cc} w_{11} & w_{12} \\  w_{21} & h \end{mat} , 
\]
where $W(s)$ is a $2\times2$ real rational transfer function in
$\mathcal{RL}_{\infty}$, introduced to represent various classes of 
uncertainties in each subsystem. In particular, agents with the additive, 
multiplicative, and feedback type perturbtions can be described by
${\cal F}_u(\delta_i, W)$ with the corresponding $W$ given respectively 
by  
\[ 
W_a = \begin{mat}{cc} 0 & w_{a} \\  1 & h \end{mat}, 
W_m = \begin{mat}{cc} 0 & w_{m} \\  h & h \end{mat},  
W_f = \begin{mat}{cc} w_f & w_{f} \\  h & h \end{mat}. 
\]

Now consider an interconnection of $n$ heterogeneous uncertain 
agents ${\cal F}_u(\delta_i, W)$, $i\in\II_n$, through the connectivity matrix $A$.
The uncertain network system can be described by
\begin{equation} \label{sys}
\begin{mat}{c} z \\ y \end{mat}=H(s)
\begin{mat}{c} w \\ u \end{mat}, \hs
\begin{array}{l} w=\Delta(s) (z+d_\delta) + d_{h1}, \\ 
u=A(y+d_a) + d_{h2}, \end{array}
\end{equation}
\begin{equation} \label{sysH}
H(s) := W(s)\otimes I_n, \hs
\Delta(s):=\diag(\delta_1,\ldots,\delta_n),
\end{equation}
where $z(t), w(t), y(t), u(t)\in\IR^n$ are internal signals in the system and 
$d_{h1}(t), d_{h2}(t), d_\delta(t), d_a(t)\in\IR^n$ are external signals
as depicted in Fig.~\ref{fig:LFT_UncertainNetwork}.  
Here, $A\in\IR^{n\times n}$ is a constant matrix representing a network structure or 
the interactions between the SISO uncertain subsystems/agents.

Following the definition in Vinnicombe's book \cite{Vinnicombe2001}, internal stability of the feedback system is defined as follows.

\begin{defn} \label{def:IS}
The feedback system 
in Fig.~\ref{fig:LFT_UncertainNetwork} is said to be internally 
stable if the $4n \times 4n$ transfer function from the 
external input signal $[d_{h1}, d_{h2}, d_\delta, d_a]$ to the output signal 
$[z, y, w, u]$ is in $\RHinf^{4n \times 4n}$. 
This transfer function is denoted by $\Sigma(\Delta,H,A)$.
\end{defn}

We make the following standing assumption.

\begin{ass} \label{ass:stand}
The nominal network system is unstable, i.e., the transfer function from
$(d_a,d_{h2})$ to $(y,u)$ in Fig.~\ref{fig:LFT_UncertainNetwork} with $\Delta=0$
is unstable, and the uncertainty is stable, 
$\Delta \in \mathcal{RH}_{\infty}^{n \times n}$.
\end{ass}

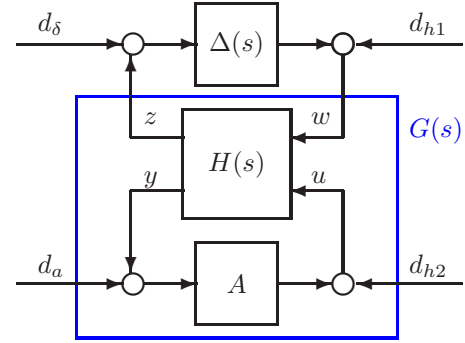
\begin{figure}[htb]
  \centering
\begin{picture}(172,129)(23,-25)
 \thicklines
 \put(75,25){\framebox(40,40){$H(s)$}}
 \put(160,4){$d_{h2}$}
 \put(20,4){$d_a$}
 \put(135,90){\circle{8}}
 \put(180,90){\vector(-1,0){40}}  
 \put(160,95){$d_{h1}$}
  \put(20,95){$d_\delta$}
 \put(56,90){\circle{8}}
 \put(12,90){\vector(1,0){40}}  
 \put(80,75){\framebox(30,30){$\Delta(s)$}}
  \put(35,-20){\color{blue}\framebox(120,90){}}
   \put(160,55){\color{blue} $G(s)$}
 \put(60,60){$z$}
 \put(55,55){\line(1,0){20}}  
 \put(55,55){\vector(0,1){31}}  
 \put(60,90){\vector(1,0){20}}  
 \put(123,60){$w$}
 \put(110,90){\vector(1,0){21}}  
 \put(135,55){\line(0,1){31}}  
 \put(135,55){\vector(-1,0){20}}  
 \put(80,-15){\framebox(30,30){$A$}}
 \put(60,38){$y$}
 \put(135,0){\circle{8}}
 \put(56,0){\circle{8}}
 \put(180,0){\vector(-1,0){40}} 
 \put(12,0){\vector(1,0){40}}  
 \put(135,90){\circle{8}}
 \put(55,35){\line(1,0){20}}  
 \put(55,35){\vector(0,-1){31}}  
 \put(60,0){\vector(1,0){20}}  
 \put(123,38){$u$} 
 \put(110,0){\vector(1,0){21}}  
 \put(135,4){\line(0,1){31}}  
 \put(135,35){\vector(-1,0){20}}  
\end{picture}
  \caption{Uncertain network system $\Sigma(\Delta,H,A)$} 
\label{fig:LFT_UncertainNetwork}
\end{figure}

\vspace{-0.5cm}

\subsection{Robust Instability Analysis Problem}
\label{subsec:ProblemFormulation}

Under Assumption~\ref{ass:stand}, we are interested in conditions for 
{\em robust instability} of the uncertain network system, i.e.,
$\Sigma(\Delta, H, A)\notin\RHinf^{4n \times 4n}$ for all $\Delta\in\bfDel_{d}$ with a norm bound,
where $\bfDel_d$ is the set of stable diagonal uncertainties defined by
\begin{equation} \label{bfDel}
\bfDel_d := \{\diag(\delta_1,\ldots,\delta_n):\,\delta_i\in\RHinf\,\}.
\end{equation}
Clearly, the robust instability condition is violated when there exists 
a $\Delta\in\bfDel_{d}$ with its norm within the bound  that makes the network system stable. Let us
define the set of stabilizing uncertainties as follows:
\begin{equation} \label{eq:SdG}
\mathbb{S}_d  :=  \{ \Delta\in\mathbf{\Delta}_d  ~|~ 
\Sigma(\Delta,H,A)\in\RHinf^{4n\times 4n} \} . 
\end{equation}
The robust instability radius (RIR), denoted by $\rho_*$, 
is then defined as the smallest magnitude of the perturbation that stabilizes 
the network system:
\begin{align} 
\label{rir}
\rho_* := {\displaystyle\inf_{\Delta\in\mathbb{S}_d}} ~\|\Delta\|_{\infty}.
\end{align}
The problem under study is about finding a method for calculating $\rho_*$ for a given (unstable) network system. 
\section{Basic Stability Analysis}
\label{sec:bsa}

This section provides stability conditions for the network system 
$\Sigma(\Delta,H,A)$ with a fixed $\Delta\in\RHinf^{n\times n}$, 
forming a foundation for the robust instability analysis with 
uncertain $\Delta\in\bfDel_d$ in the sections that follow.

\subsection{General Case}
\label{subsec:Stability_general}

Under Assumption 1, we have the following simplified internal stability condition
because we can delete the external input signals inserted to the stable blocks 
($d_\delta$ and $d_a$) and the output signals of the stable blocks  ($w$ and $u$).

\begin{prop} \label{prop:IS1} 
Suppose $\Delta$ is stable. Then the network system $\Sigma(\Delta, H, A)$ 
is internally stable if and only if the $2n \times 2n$ transfer function 
from the external input $[d_{h1}, d_{h2}]$ to the output $[z, y]$
is stable, i.e.,
\begin{equation} \label{eq:P}
P(s) := (I - H(s)\nabla(s))^{-1}H(s)  \in \RHinf^{2n \times 2n} , 
\end{equation} 
where 
\begin{equation} 
\nabla(s) := \begin{bmatrix} \Delta(s)  & 0 \\ 0 & A \end{bmatrix}
\in \RHinf^{2n \times 2n} . 
\end{equation} 
\end{prop}
\begin{proof} 
First note that internal stability of $\Sigma(\Delta, H, A)$ is equivalent to the stability of the feedback interconnection of $H$ and $\nabla$.
Since $\nabla$ is in $\RHinf$, application of Corollary 5.2 in \cite{zhou1998} 
yields that $P \in \RHinf$ is a necessary and sufficient condition 
for the internal stability of $\Sigma(\Delta, H, A)$. 
\end{proof} 

\subsection{Stable Agent Case}

We consider the case where the perturbed agents are stable by imposing the following requirement.

\begin{ass} \label{ass:stable}
$H\in\mathcal{RH}_{\infty}$, 
${\cal F}_u(\Delta,H)$ is internally stable.
\end{ass}
Note that $\|w_{11}\|_{\infty} < 1/\|\delta\|_{\infty}$ guarantees the 
stability of the perturbed agent ${\cal F}_u(\delta_i,W)$ by the small gain
theorem, provided $W$ and $\Delta$ are stable.

The network system $\Sigma(\Delta,H,A)$ can be viewed as a feedback system 
consisting of $\Delta$ and the nominal network block  
\begin{equation} \label{eq:Gs}
G(s):=\LFT(H(s),A)  
\end{equation}  
as shown in Fig.~\ref{fig:LFT_UncertainNetwork}, where 
$\LFT(H,A)$ denotes the lower linear fractional transformation of $H$ and $A$: i.e., 
\[
z=\LFT(H,A)w \hs \Leftrightarrow \hs
\begin{mat}{c} z \\ y \end{mat}=H
\begin{mat}{c} w \\ u \end{mat}, \hs
u=Ay.
\]
Note that $G(s)$ can be written as
\[
\begin{array}{l}
G(s)=N(s)M(s)^{-1}, \\ N(s):=w_{11}(I-hA)+w_{12}w_{21}A, \; M(s):=I-hA.
\end{array}
\]
The idea is that the network system $\Sigma(\Delta,H,A)$ is internally stable 
when the feedback system $(\Delta,G)$ is internally stable and $G(s)$ has no
inherent unstable pole/zero cancellation, i.e., $M(s)$ and $N(s)$ have no
pole/zero cancellation in the closed right half plane. The following result
makes a formal statement.

\begin{prop} \label{prop:stab}
Consider the network system $\Sigma(\Delta,H,A)$ in 
Fig.~\ref{fig:LFT_UncertainNetwork}. Suppose Assumptions~\ref{ass:stand} and
\ref{ass:stable} hold. The network system is internally stable if and only if
\begin{align} 
& \label{eq:IGIDelta}
\begin{bmatrix} I_n \\ G(s) \end{bmatrix}
(I - \Delta(s) G(s))^{-1}  \in \RHinf^{2n \times n} \\
& \label{Vrank}
\rank \; V(s) = n, \hs \forall~\Re(s) \geq 0,
\end{align} 
hold, where
\[
V(s) := \begin{bmatrix} I_n  - h(s)A \\ w_{12}(s)w_{21}(s) A \end{bmatrix} 
\]
Moreover, (\ref{eq:IGIDelta}) is necessary with no assumptions.
\end{prop}
\begin{proof}
We first show necessity. Note that internal stability of the network system
implies stability of the transfer function from 
$d_{h1}$ to $(w,z)$, which is 
the condition in (\ref{eq:IGIDelta}). To show the necessity of (\ref{Vrank}),
suppose $V(s)$ is rank deficient for some $s=s_o\in\IC$ in the right half plane.
Then with $L(s):=w_{12}(s)w_{21}(s)A$, there exists nonzero $v_o\in\IC^n$ such that 
\[
M(s_o)v_o=0 \quad \text{and} \quad \hs L(s_o)v_o=0, 
\]
which implies that the closed-loop transfer function 
\[
y=\big((M-\hat\Delta L)^{-1}-I\big)d_a, \hs \hat\Delta:=\Delta(I-w_{11}\Delta)^{-1},
\]
is unstable with a pole at $s_o$. Hence, (\ref{Vrank}) is a necessary condition
for internal stability of $\Sigma(\Delta,H,A)$.

To prove sufficiency, suppose (\ref{eq:IGIDelta}) and (\ref{Vrank}) hold.
First note that $M(s)$ and $N(s)$ are coprime due to (\ref{Vrank}) because
\[
V = \begin{mat}{c} I-hA \\ N-w_{11}(I-hA) \end{mat} 
= \begin{mat}{c} M \\ N-w_{11}M \end{mat}.\]
With $G=NM^{-1}$, we see that (\ref{eq:IGIDelta}) implies 
\[
\begin{mat}{c} M \\ N \end{mat} (M-\Delta N)^{-1}\in\RHinf^{2n\times n},
\]
which in turn implies 
$F := (M-\Delta N)^{-1} \in \RHinf^{n \times n}$
due to coprimeness of $M$ and $N$. Lemma~\ref{lemma:IS2} in the Appendix then 
guarantees that the network system $\Sigma(\Delta,H,A)$ is internally stable. 
\end{proof}

We here provide a simple example that shows that  Assumption~\ref{ass:stable}
is important in Proposition~\ref{prop:stab}.
Consider 
\[ 
\begin{array}{l} n=2, \\ 
\delta_1 = \delta_2 = 3/4, \end{array} ~
h(s) = \frac{1}{3-s}, ~ 
A = \begin{bmatrix}
    1 & 1 \\
    1 & 1
\end{bmatrix}, ~
W = \begin{bmatrix}
0 & 1 \\
h & h
\end{bmatrix},
\]
for which 
$\mathcal{F}_u(\Delta_, W) = (1+\frac{3}{4}) h I_2$ 
models the perturbed system  
and $G(s) := \mathcal{F}_\ell(H(s), A) = \frac{1}{1-s}A$
is unstable. Evidently, Assumption~\ref{ass:stand} holds but 
Assumption~\ref{ass:stable} is violated. 
It is straightforward to verify that (\ref{eq:IGIDelta}) 
holds, and (\ref{Vrank}) is also satisfied for all $s$
in CRHP except at the pole $s=3$ of $h(s)$. 
However, (\ref{eq:P}) in Proposition~\ref{prop:IS1} is violated. 
Specifically, every entry in the bottom two rows of $P(s)$ 
in (\ref{eq:P}) has an unstable pole located at $3$. 
In other words, the overall network is unstable by Proposition~\ref{prop:IS1}.
Thus, conditions (\ref{eq:IGIDelta}) and (\ref{Vrank}) do not guarantee internal
stability of the network system unless Assumption~\ref{ass:stable} holds. 

\begin{remark}
\label{rmk:assump2} 
Condition \eqref{Vrank} is not restrictive in practice for the following reasons.
\begin{itemize}
\item
Condition \eqref{Vrank} holds if $w_{12}(s)w_{21}(s)$ is minimum phase. 
Hence, a natural outer weight function $w_a(s)$ in the additive perturbation 
case satisfies it. 
\item
Suppose $w_{12}(s)w_{21}(s)$ is a product of $h(s)$ and an outer function $w_3(s)$. 
Observe that condition \eqref{Vrank} holds even if $h(s)$ has nonminimum phase zeros. Hence, a natural outer function $w_m(s)$ (resp. $w_f(s)$) in the multiplicative (resp. feedback) perturbation case would satisfy the condition.
\item
In general, when condition \eqref{Vrank} is violated, arbitrarily small perturbations on $A$ and/or $w_{12}(s)w_{21}(s)$ can restore its satisfaction. 
\end{itemize}
Thus, internal stability of the network system is essentially equivalent to 
(\ref{eq:IGIDelta}) and we will focus on this condition in our further analysis. 
\end{remark}

\section{Upper and Lower Bounds on RIR}
\label{sec:Preliminaries}

This section is concerned with upper and lower bounds on the RIR as preliminary 
analysis. First note that it is very difficult to get any upper or lower bounds 
on the RIR based on Proposition~\ref{prop:IS1} since $P(s)$ is a function of 
$\Delta(s)$ as well as $H(s)$ and $A$. On the other hand, we may expect that some 
upper or lower bounds can be derived from Proposition~\ref{prop:stab} since
the network is recognized as a feedback system of $G(s)$ and 
$\Delta(s)$. This framework recasts the RIR computation as the search
for a minimum-norm perturbation $\Delta\in\RHinf^{n\times n}$ to stabilize
$G(s)$. 
The following MIMO version of the parity interlacing property 
(PIP)~\cite{Sagar1985} is important:  

\noindent
{\bf [MIMO PIP]} \; 
A general MIMO transfer function $G(s)\in\mathcal{R}_p$ is strongly stabilizable 
if and only if the number of unstable real poles of $G$ (counted 
according to their McMillan degree) between any pair of real blocking zeros 
of $G$ in CRHP $\cup \, \{ \infty \}$ is even. 

If the PIP condition is violated, then there is no stable perturbation 
$\Delta(s)$ stabilizes the network system, and therefore the RIR is infinite. 
Consequently, the analysis below will focus on the characterization of 
the RIR for $G(s)$ that passes the PIP test.

\subsection{Obvious Bounds from Proposition~\ref{prop:stab}}

We introduce the following two sets of $\Delta$, namely the sets of 
homogeneous uncertainties and full block (or unstructured) uncertainties 
that stabilize the network system $\Sigma(\Delta,H,A)$ in order to derive 
upper/lower bounds of $\rho_*$:
\begin{align*}
\mathbb{S}_h & :=  \{ \Delta=\delta I_n, \  \delta\in\RHinf  ~|~ P \in\RHinf^{2n\times 2n}  \}, \\
\mathbb{S}_f & :=  \{ \Delta\in\RHinf^{n\times n} ~|~ P \in\RHinf^{2n\times 2n} \},
\end{align*}
where $P$ is as defined in \eqref{eq:P}.
Since $\mathbb{S}_h\subseteq\mathbb{S}_d\subseteq\mathbb{S}_f$, 
an upper bound and a lower bound of $\rho_*$ are respectively given by 
\begin{align} 
\rho_* \leq & \rho_h := \displaystyle\inf_{\Delta\in\mathbb{S}_h} ~\|\Delta\|_{\infty} \label{hrir} , \\
\rho_* \geq & \rho_f := \displaystyle\inf_{\Delta\in\mathbb{S}_f} ~\|\Delta\|_{\infty} \label{urir} , 
\end{align} 
i.e., we have $\rho_f \le \rho_* \le \rho_h$. 

We also introduce three additional sets of $\Delta$ for the case of stable 
perturbed agents in order to have computable bounds  based on 
Proposition~\ref{prop:stab}. 
In particular, we define the sets of 
$\Delta$, which stabilizes $G(s)$, as follows:
\begin{align*}
\mathbb{S}_f^G  & :=  \{ \Delta \in \  \RHinf^{n\times n} ~|~ 
\eqref{eq:IGIDelta} \; \mbox{holds.} \} \supseteq \mathbb{S}_f. \\
\mathbb{S}_d^G & :=  \{ \Delta \in \mathbb{S}_f^G  ~|~ \Delta\in\mathbf{\Delta}_d \} 
\supseteq \mathbb{S}_d, \\ 
\mathbb{S}_h^G & :=  \{ \Delta \in \mathbb{S}_f^G  ~|~ \Delta=\delta I_n, \}  \supseteq \mathbb{S}_h  . 
\end{align*}
It is clear that 
$\rho_h^G  \ge \rho_*^G \ge \rho_f^G \ge \|G\|_{\infty}^{-1}$
holds, where 
\begin{align} 
& \rho_*^G := {\displaystyle\inf_{\Delta\in\mathbb{S}_d^G}} ~\|\Delta\|_{\infty} 
\leq \rho_*,   \\ 
&\rho_h^G := \displaystyle\inf_{\Delta\in\mathbb{S}_h^G} ~\|\Delta\|_{\infty} \label{hrir_G} 
\leq \rho_h,   \\ 
&\rho_f^G := \displaystyle\inf_{\Delta\in\mathbb{S}_f^G} ~\|\Delta\|_{\infty} \label{urir_G} 
\leq \rho_f. 
\end{align}
Note that $\rho_f^G \ge \|G\|_{\infty}^{-1}$ holds by the small-gain theorem.
We can see from Proposition~\ref{prop:stab} that equalities hold in the above 
inequalities, i.e., $\rho_* = \rho_*^G$, $\rho_h = \rho_h^G$, and $\rho_f = \rho_f^G$
when the uncertain agents are stable and there is no inherent unstable pole/zero
cancellation within $G(s)$.

\subsection{Comprehensive Results}

This subsection derives upper and lower bounds on the RIR by exploiting a system decomposition. 
To this end, 
we assume for simplicity the following to avoid naive situations:
\begin{ass} \label{ass:Adiag}
$A \in \IR^{n \times n}$ is diagonalizable.
\end{ass}

Let a spectral decomposition of $A$ be given by
\begin{equation} \label{eq:diagA}
A=T\Lambda T^{-1}, \hs \Lambda=\diag(\lambda_1,\ldots,\lambda_n).
\end{equation}
Then, the stability of $\Sigma(\Delta, H, A)$ is equivalent to that of $\Sigma(T^{-1}\Delta T, H, \Lambda) = \Sigma(\hat{\Delta}, \hat{G})$. By Proposition~\ref{prop:stab}, we may investigate the stability of the feedback interconnection of $\hat{\Delta} := T^{-1}\Delta T$ and $\hat{G} := \LFT(H, \Lambda)$. 
It is clear that $\hat{G}$ is a diagonal matrix of which the diagonal elements are represented by  
\begin{equation} \label{eq:gi}
g_i(s) := \LFT(W(s),\lambda_i), \hs i \in \II_n . 
\end{equation}
For the cases of the additive, the multiplicative, and the feedback type perturbations with 
an eigenvalue $\lambda$, they are  respectively given by 
setting $\lambda:=\lambda_i$ in
\[
g_a(s) = \frac{w_a \lambda}{1 - \lambda h}, \hs 
g_m(s) = \frac{w_m \lambda h}{1 - \lambda h}, \hs 
g_f(s) = \frac{w_f}{1 - \lambda h}. 
\] 

A summary of the results derived in this subsection is given as follows. 
\begin{theorem} \label{prop:ULsummary} 
Suppose that the uncertain network system $\Sigma(\Delta, H, A)$ satisfies 
Assumptions~\ref{ass:stand} and \ref{ass:Adiag}.
Let $\IU\subseteq\II_n$ denote the set of indices $i$ such that $g_i(s)$ is unstable. 
Then, we have 
\begin{equation} \label{ULsummary}
\begin{array}{cccccccccc}
&& \rho_f &\le& \rho_* &\le& \rho_h &&  \vspace{2mm} \\
&& \rotatebox{90}{$\le$}^s && \rotatebox{90}{$\le$}^s && \rotatebox{90}{$\le$}^s && \vspace{2mm} \\
\|G\|_{\infty}^{-1}  &\le& \rho_f^G &\le& \rho_*^G &\le& \rho_h^G &=& 
\rhoU \vspace{2mm} \\
&& \rotatebox{90}{$\le$}_n && \rotatebox{90}{$\le$}_{dn} && \rotatebox{90}{$\le$}&&  \\
&& \rhoL_p && \rhoL_p && \rhoL_+  && 
\end{array}
\end{equation}
where
\begin{align}
& \rhoU := \!\! \inf_{\delta \in \RHinf} \{\|\delta\|_{\infty}:  
\frac{\begin{mat}{cc} 1 & g_i \end{mat}}{1-\delta g_i}\in\RHinf, ~
\forall i \in \II_n \},  \label{rhou} \\
& \rhoL_p := \min_{i\in\II_n}~ 1/\|g_i\|_{\infty}, \label{rhol} \\
& \rhoL_+:=\max_{i\in\IU} 1/\|g_i\|_{\infty},   \label{rho+}
\end{align}
and the symbols $\rotatebox{90}{$\le$}^s$, $\rotatebox{90}{$\le$}_n$, and 
$\rotatebox{90}{$\le$}_{dn}$ mean the following:
\begin{itemize} 
\item
$\rotatebox{90}{$\le$}^s$ means that the inequality holds in general and that
it becomes an equality if every possible perturbation $\Delta$ satisfies Assumption~\ref{ass:stable}  
and~\eqref{Vrank} holds.
\item
$\rotatebox{90}{$\le$}_n$ means that this inequality holds when $A$ is normal,
in which case we have $\rhoL_p = \|G\|_{\infty}^{-1}$.  
\item
$\rotatebox{90}{$\le$}_{dn}$ means that this inequality holds when $A$ is diagonally 
normalizable, i.e., $A\in\Rset_{DN}^{n\times n}$ with 
\begin{eqnarray}
\hspace*{-8mm}
\Rset_{DN}^{n \times n} \!\!\! & := \!\!\! & \{ A \in \Rset^{n \times n} \; | \: 
~\exists \; \mbox{diagonal matrix } D > 0  
\nonumber \\ 
&& \mbox{such that} \; 
DAD^{-1} \; \mbox{is normal} \; \} 
\end{eqnarray} 
in which case we have $\rhoL_p = \|DGD^{-1}\|_{\infty}^{-1} $. 
\end{itemize}
\end{theorem}
\begin{proof}
The four inequalities $\leq$ in the middle follow from the inclusion 
relationship between the full block, diagonal, and homogeneous uncertainty sets.
The inequalities $\rotatebox{90}{$\le$}^s$ are simple consequences of 
Proposition~\ref{prop:stab}. $\|G\|_{\infty}^{-1}  \le \rho_f^G$ is from the small gain 
argument.  $\rho_h^G = \rhoU$ is also trivial, since $\rhoU$ gives the smallest norm of $\delta(s)$ 
which simultaneously marginally stabilizes all $g_i(s)$.  
$\rhoL_+ \leq \rho_h^G$ can be shown by contradiction.  Suppose $\rhoL_+ >  \rho_h^G$.  
Then, it is clear that there exists at least one unstable $g_i$ which cannot be stabilized.  
This contradicts the assumption. 
The inequality between $\rho_f^G$ and $\rhoL_p$ is from Inoue {\em et al}. 
It is also a special case of the inequality between $\rho_*^G$ and $\rhoL_p$ that 
corresponds to the case where $A$ is diagonally normalizable, i.e.,
we have $UDAD^{-1}U^* = \diag \{ \lambda_i \}$.  
Then the stability of the feedback loop system consisting of $\Delta(s)$ and $G(s)$ is equivalent to that 
consisting of $UD \Delta(s) D^{-1}U^*$ and $\diag \{ g_i(s) \}$ for some diagonal $D$ and unitary $U$. 
Since $D \Delta(s) D^{-1} = \Delta(s)$, we have 
\[
 \| UD \Delta D^{-1}U^* \|_\infty = \| U \Delta U^*| \|_\infty =  \|\Delta| \|_\infty .
\]
By the small gain argument we can see that the feedback loop system cannot stabilized 
by any such $\Delta(s)$ if 
\[
\|\Delta \|_\infty < \rhoL_p = \min 1/\| g_i\|_\infty 
= (\max \| g_i\|_\infty)^{-1} 
\]
holds.  This implies that  $\rho_*^G \ge \rhoL_p$. 
\end{proof}

\begin{remark} 
We have two remarks on Theorem~\ref{prop:ULsummary}. 
\begin{itemize}
\item
Suppose Assumptions~\ref{ass:stand}--\ref{ass:Adiag} hold. Then we have 
\[
\rhoL_+ \le \rho_h^G = \rho_h = \rhoU  . 
\]
This is a nice result, but computation of $\rhoU$ is very difficult in general. 
The reason is that the corresponding problem is a minimum norm  simultaneous 
strong stabilization problem. 
\item
$\rhoL_p$ possibly improves the lower bound $\|G\|_{\infty}^{-1}$ 
if $A$ is diagonally normalizable but not normal.
\end{itemize}
\end{remark}

Now compare $\rhoL_p$ (a lower bound for the case of heterogeneous perturbation) and 
$\rhoL_+$ (a lower bound for the case of homogeneous perturbation). There are two differences: $\rhoL_p$ is {\em min} among all $i \in\II_n$, but 
$\rhoL_+$ is {\em max} among not all but only $i \in\IU$. 
This clearly shows that $\rhoL_p \leq \rhoL_+$, i.e., 
 $\rhoL_+$ gives a better lower bound than  $\rhoL_p$ 
if $\rho_* = \rho_h$ holds. 
This property is confirmed by a numerical example related to 
a class of biological systems \cite{hara:19, hara2022instability}. 

Consider a set of uncertain $n$ agents $(1+\delta_i(s))h(s)$, $i\in\II_n$, 
with cyclic connections. This is a special case of (\ref{sys}) with 
\begin{equation} \label{cyc}
W(s) = \begin{mat}{cc} 0 & h(s) \\ 1 & h(s) \end{mat}, \hs
A=\begin{mat}{cc} o\t & -\mu \\ -\mu I & o \end{mat} , 
\end{equation}
where $\mu\in\IR$ is a positive scalar and $o\in\IR^{n-1}$ is a zero vector. 
The nominal dynamics of the agent are represented by 
$h(s) = c\beta/\{ (s+a)(s+b) \}$.  
The parameters are set as $a = 5.1986,$, $b = 0.4612$, $c = 30$ $\beta = 24$, and $\mu = 0.006247$, 
which are from a practical example of gene regulatory networks 
investigated in \cite{NetworkRIR}.  
Note that $A$ is a cyclic matrix with all non-positive entries. 
Hence, $A \in \Rset_{DN}$ \cite{hara:19} and 
the eigenvalues of A are located at $\lambda_k = \mu e^{j(k/n)\pi}$ for $k=1,3,\ldots,2n-1$.  
Since $h(s)$ is stable, by Theorem~\ref{prop:ULsummary}
or the inequalities in (\ref{ULsummary}), upper and lower bounds on the
robust instability radius $\rho_* = \rho_*^G$ are given by $\rhoU$ and $\rhoL_p$, respectively. 

Numerical results of robust instability radius for odd $n$ are shown in Fig.~\ref{fig:rir}. 
The lower bounds $\rhoL_p$ plotted by the blue curve are always far below 
the red plots indicating $\rhoL_+$ except $n=3$ or $n=5$.    
The reason is as follows. 
$\rhoL_+$ are based on the eigenvalue $\lambda_1$ that is the closest to the real axis, 
while $\rhoL_p$ is based on the eigenvalue $\lambda_k$ closest to 
the inverse Nyquist $1/h(j\omega)$.  They are coincident only when $n=3$, 
and the gaps are big for other cases in general.  
This example clearly shows that the small gain argument does not work so well 
for the exact robust instability analysis in general. 

It should be emphasized that for a class of network systems 
with cyclic graph having stable agents, $\rhoL_+$ is exactly equal to the 
RIR \cite{NetworkRIR}, where the logic is based on the monotone and the 
convexity properties of the inverse Nyquist plots $1/h(j\omega)$.

\begin{figure}[h]
\hspace{8mm}
\epsfig{file=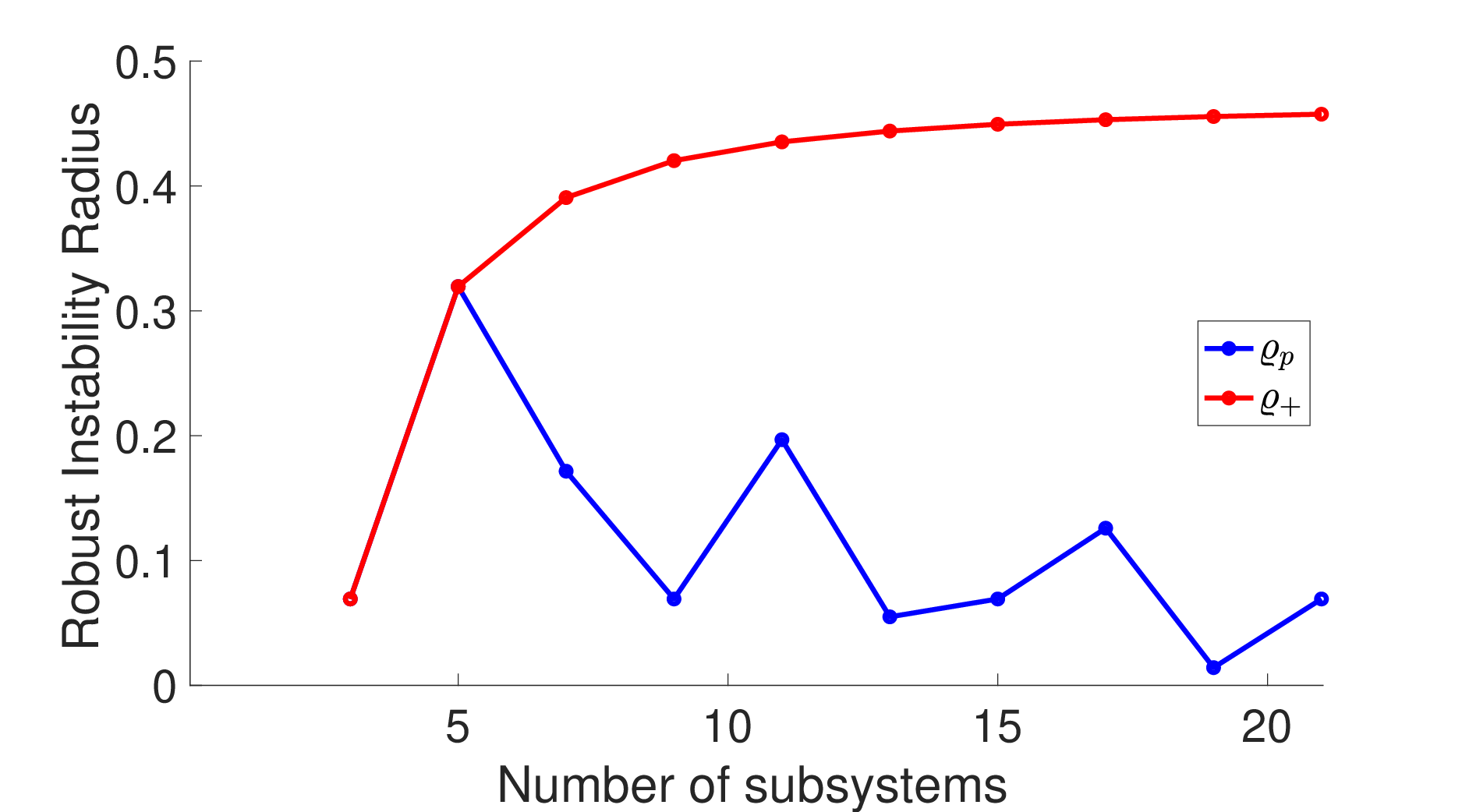,width=70mm}
\caption{Robust instability radius. Blue and red plots are $\rhoL_p$ and $\rhoL_+$ 
defined by (\ref{rhol}) and by (\ref{rho+}), respectively. }
\label{fig:rir}
\end{figure}
\section{Rank Deficient Case}
\label{sec:RankDeficient}

This section investigates the robust instability radius for the case where 
the interaction matrix $A$ is rank deficient, i.e., $\rank(A) =: k < n$. 

\subsection{General Analysis}

We have the following lemma regarding a necessary condition for stability 
of the network. 

\begin{lemma}
\label{lem:nec-rankdef}
Suppose $w_{11},\Delta\in\RHinf$, 
$A$ is rank-deficient with semisimple eigenvalue(s) at zero, and 
$\Sigma(\Delta,H,A)$ is internally stable. 
Then $H\in\RHinf$. 
\end{lemma}
\begin{proof}
See Appendix~\ref{prof:nec-rankdef}. 
\end{proof}

The above result indicates that, when the network connections are rank deficient,
stability of $H(s)$ is a necessity under a mild condition. For the rest of this 
section, we consider the case where the agent dynamics are 
modeled by a common nominal transfer function $h(s)$ with 
heterogeneous multiplicative uncertainties as described in 
Section~\ref{subsec:UncertainDNS}. In this case, $W_m$ has $w_{11}=0$, 
and in light of Lemma~\ref{lem:nec-rankdef}, we require 
$h, \ w_m\in\RHinf$ without loss of generality in our analysis 
since otherwise the RIR is infinite. 
We let $w_m$ be an outer function. 

Under Assumption~\ref{ass:Adiag}, let $A$ be decomposed as 
$A=T\Lambda T^{-1}, \; \Lambda = \mathrm{diag}(\Lambda_k,0_{(n-k)})$, 
where $T$ is a nonsingular matrix and 
$\Lambda_k = \mathrm{diag}(\lambda_1,\lambda_2,\cdots,\lambda_k), 
\; \lambda_i\not=0, \: i\in\II_k.$
The decomposition can also be seen as $A=L\Lambda_k R$, 
where $L\in\mathbb{C}^{n\times k}$ and $R\in\mathbb{C}^{k\times n}$ satisfy $RL=I_k$. 
Then we have 
\[
\mathcal{F}_\ell(H,A)=T\mathcal{F}_\ell(H,\Lambda) T^{-1}=LG_kR 
\]
\[
G_k:= \diag(g_1,\cdots, g_k), \hs
 g_i = \frac{\lambda_i w_m h}{1-\lambda_i h}, \hs i\in\II_k.
\]
Hence, the network system essentially reduces to the feedback system of
$\Delta_k:=R\Delta L$ and $G_k$. With this insight, we have the following result. 

\begin{lemma} \label{lem:stbGk}
Consider the network system $\Sigma(\Delta,H,A)$ in 
Fig.~\ref{fig:LFT_UncertainNetwork} with multiplicative perturbations, i.e., $W=W_m$.  
Suppose $w_m(s)$ is an outer function, Assumptions~\ref{ass:stand}--\ref{ass:Adiag}
hold, and $A$ is rank deficient.
Define $\Delta_k(s)$ and $G_k(s)$ as above through a spectral 
decomposition of $A$. Then, the network system is internally stable if and only if
\begin{align}
\begin{bmatrix} I_k \\ G_k \end{bmatrix}(I_k-\Delta_k G_k)^{-1}
\in\RHinf^{2k\times k}.
\label{eq:lowrank_S3}
\end{align}
\end{lemma}

\begin{proof}
Applying  Lemma~\ref{lemma:IS2} in the Appendix, $\Sigma(\Delta,H,A)$ is stable if and only if 
the corresponding $F(s)$ defined by \eqref{eq:F} is stable. 
In this case, $F\in\RHinf^{n\times n}$ holds if and only if the following 
determinant is nonzero for all $s\in$ CRHP: 
\begin{align*} 
\det\big(F(s)^{-1}\big)
=  & \det(I_n-h(s)A-w_m(s)h(s)\Delta(s) A) \\
= & \det(I_n-h(s)(I+w_m(s)\Delta(s))L\Lambda_k R)\\
= & \det(I_k- h(s)\Lambda_k - w_m(s)h(s)\Delta_k(s)\Lambda_k).
\end{align*}
Noting that the first and third determinants have analogous expressions, 
consider the reduced network system $\Sigma(\Delta_k, H_k, \Lambda_k)$
where $H_k:=W_m\otimes I_k$. 
It can be verified that the reduced system 
satisfies Assumptions~\ref{ass:stand} and \ref{ass:stable}. \linebreak
Then Lemma~\ref{lemma:IS2}
and the development above yield that $\Sigma(\Delta,H,A)$ is internally stable 
if and only if $\Sigma(\Delta_k, H_k, \Lambda_k)$ is  internally stable.
Since $w_m$ is an outer function, by Remark~\ref{rmk:assump2} and 
Proposition~\ref{prop:stab}, %
$\Sigma(\Delta_k, H_k, \Lambda_k)$ is 
internally stable if and only if \eqref{eq:lowrank_S3} holds. 
\end{proof}

\subsection{Two-Step Optimization Procedure}

Now consider the following problems concerning the feedback interconnection of $G_k$ and $\Delta_k$:
\begin{align}
&  \inf_{\Delta} \|\Delta\|_{\infty}
\mbox{ \ \ s.t. \ }  R\Delta L\in \mathbb{S}_f^{G_k}, \  \Delta\in\mathbf{\Delta}_d
\label{eqiv:rhod_G} \\
&  \inf_{\Delta_k} \|\Delta_k\|_{\infty} 
\mbox{ \ s.t. \ }  \Delta_k \in \mathbb{S}_f^{G_k} 
\label{2steps:1} \\
&  \inf_{\Delta} \|\Delta\|_{\infty}
\mbox{ \ \ s.t. \ } R\Delta L = \Delta_k^*, \ 
\Delta\in\mathbf{\Delta}_d 
\label{2steps:2}
\end{align}
where $\Delta_k^*$ is an optimizer of \eqref{2steps:1}, and the set
$\mathbb{S}_f^{G_k}$ is defined similarly to $\mathbb{S}_f^G$. 

Note that problems~\eqref{2steps:1} and~\eqref{2steps:2} shall be bundled together, as 
the constraint in~\eqref{2steps:2} depends on the solution of~\eqref{2steps:1}. 
Also note that in problem~\eqref{2steps:1}, we do not require any structure for 
$\Delta_k$. 

The following proposition shows that the RIR of the rank deficient network 
$\Sigma(\Delta,H,A)$ can be found by solving~\eqref{eqiv:rhod_G}. 
\begin{prop}
\label{prop:red}
Under the setup of Lemma~\ref{lem:stbGk}, the infimum of 
\eqref{eqiv:rhod_G} is equal to $\rho_*$.
\end{prop}
\begin{proof}
The result follows directly from Lemma~\ref{lem:stbGk}.
\end{proof}

Unfortunately, it is difficult to solve (\ref{eqiv:rhod_G}) directly. 
The idea we propose here is that, instead of solving~\eqref{eqiv:rhod_G} for 
$\rho_*$, we solve the bundled problems~\eqref{2steps:1} and~\eqref{2steps:2}.  
We refer to this approach as a ``two-step procedure.''
Note that the infimum of the two-step procedure gives an \emph{upper bound} 
on $\rho_*$. 
To see this, notice that the argument of infimum of 
problem~\eqref{2steps:2} satisfies the constraint in problem~\eqref{eqiv:rhod_G}. 
The following proposition shows that one does acquire an optimal solution 
to~\eqref{eqiv:rhod_G} by the two-step procedure under a certain condition.  

\begin{prop} \label{prop:two_step_v2} 
Suppose that $A$ is diagonally normalizable. 
Then, the infimum of problem~\eqref{eqiv:rhod_G} is equal to that of
problem~\eqref{2steps:2} if the following condition holds: 
An optimizer $\Delta_k^*$ of problem~\eqref{2steps:1}
is diagonal and homogeneous; i.e. $\Delta_k^* = \delta^*I_k$ for some
$\delta^*\in\RHinf$. In this case, 
an optimal solution to problem~\eqref{eqiv:rhod_G} is 
given by $\Delta=\delta^*I_n$ and all the infimum values of 
(\ref{eqiv:rhod_G})--(\ref{2steps:2}) are the same. 
\end{prop}

\begin{proof}
We will show that under the supposition and the condition in the proposition, 
the infimum in~\eqref{2steps:2} gives a lower bound
on that in~\eqref{eqiv:rhod_G}, and hence they must be equal. 
Let $\Delta^*$ and $\Delta_k^*$ denote the 
arguments of infimum of problems~\eqref{eqiv:rhod_G} and~\eqref{2steps:1}, respectively. 
Since $R\Delta^*L$ is a feasible solution to problem~\eqref{2steps:1}, we have 
\begin{align*}
\|{\Delta_k^*}\|_{\infty}& = \|R\Delta^*L\|_{\infty}  
= \|RW\Delta^* W^{-1}L\|_{\infty}\\  
&\le\|RW\|\|\Delta^*\|_{\infty}\|W^{-1}L\|   \le \|\Delta^*\|_{\infty} , 
\end{align*}
where $W$ is an invertible diagonal matrix. 
Here, we used the following two facts:  
(i) $\Delta^*$ being diagonal implies $W\Delta^* = \Delta^* W$ and 
(ii) the diagonal normalizability of $A$ leads to $\|RW\|\cdot\|W^{-1}L\| \le 1$ for some $W$ by Lemma~\ref{lem:1} in the Appendix. 

Furthermore, if $\Delta_k^*=\delta^* I_k$, 
then $\Delta := \delta^*I_n$ is a solution to~\eqref{2steps:2} 
since $RL=I_k$. As such, the infimum of~\eqref{2steps:2} 
is less than or equal to $\|{\delta^*}\|_{\infty} 
=\|\Delta_k^*\|_{\infty}\le\|\Delta^*\|_{\infty}$. 
\end{proof}

If $A$ has rank one, then the condition in Proposition~\ref{prop:two_step_v2} is
satisfied because $G_k(s)$ is a scalar transfer function and so is the solution
$\Delta_k^*(s)$ to problem~\eqref{2steps:1}. For the scalar case, the infimum of
problem~\eqref{2steps:1} is fairly well characterized in \cite{hara2023exact}.
Together with this result and the characterization of diagonally normalizable
rank-one matrices in Lemma~\ref{lem:dnrom}, we have the following.

\begin{theorem} \label{thm:dnrom}
Consider the network system $\Sigma(\Delta,H,A)$ with multiplicative perturbations,
i.e., $W=W_m$. Suppose $w_m(s)$ is an outer function, 
Assumptions~\ref{ass:stand}--\ref{ass:Adiag} hold, and $A$ has rank one and all
the nonzero diagonal entries have the same sign. Define $g:=w_m\lambda h/(1-\lambda h)$ where 
$\lambda\in\IR$ is the nonzero eigenvalue of $A$. Define 
$\theta(\omega):=\angle g(j\omega)$, and let $\theta'$ be the derivative of $\theta$.
Then the following hold.
\begin{itemize}
\item[(a)] Suppose $|g(j\omega)|$ takes a unique 
global peak at $\omega=0$. Then $\rho^*=1/\|g\|_\infty$ if 
$g(s)$ has one unstable pole and $\theta'(0)>0$. 
\item[(b)] Suppose $|g(j\omega)|$ takes a unique
global peak at a nonzero frequency $\omega_p$. Then $\rho^*=1/\|g\|_\infty$ if
$g(s)$ has two unstable poles and
$\theta'(\omega_p)>|\sin(\theta(\omega_p))/\omega_p|$. 
\end{itemize}
In each statement, if the condition does not hold, then there exists no stable $\delta(s)$ such that $\|\delta\|_\infty\leq1/\|g\|_\infty$ and $1=\delta(s)g(s)$ is marginally stable.
\end{theorem}
\begin{proof}
From Lemma~\ref{lem:dnrom} in the Appendix, $A$ is diagonally normalizable.
Hence Proposition~\ref{prop:two_step_v2} applies.
Since $A$ is rank one, $G_k$ is a scalar transfer function and so is the optimal 
solution $\Delta_k^*$. Therefore, the infimum of~\eqref{2steps:2}, $\rho^*$,
is equal to the infimum of problem~\eqref{2steps:1}. Then the statements follow 
from the scalar-case results in \cite{hara2023exact}.
\end{proof}

For robust {\em stability} analysis, diagonal normalizability of the network graph implies that the minimum-norm destabilizing perturbation is homogeneous \cite{hara:19}, making the analysis tractable. In contrast, for robust {\em instability} analysis, the diagonal normalizability does not in general imply that the minimum-norm stabilizing perturbation is homogeneous, and hence an additional rank-one property is exploited in the above theorem.

Further results on the RIR are obtained in \cite{NetworkRIR} for the following two cases. 

\begin{itemize}
\item
$A$ is rank-one: The result can be shown by different arguments than
Theorem~\ref{thm:dnrom} that do not require the diagonal normalizability of $A$. 
\item 
$A$ is diagonally normalizable and rank-two with a pair of complex conjugate 
eigenvalues, for which the condition in 
Proposition~\ref{prop:two_step_v2} always holds. 
\end{itemize}

\section{Conclusion}
\label{sec:Concl}

We have considered a class of uncertain multi-agent systems described by 
a feedback connection of the uncertainty $\Delta(s)$, nominal agent dynamics
$H(s)$, and network connectivity matrix $A$, as in 
Fig.~\ref{fig:LFT_UncertainNetwork}. The internal stability of the 
network system is shown to be equivalent to the internal stability of
the feedback system of $\Delta(s)$ and $G(s)$ under a mild condition, where
$G(s)$ is the nominal network transfer function. The characterization of 
the robust instability radius is then reduced to strong stabilization of
$G(s)$ by $\Delta(s)$, whereby a table of upper and lower bounds on the
RIR is obtained (Theorem~\ref{prop:ULsummary}). In the rank deficient case,
the nominal agent stability is shown to be a necessity for the internal
stability of the network, and a two-step procedure is proposed for the RIR
computation. When the network interaction matrix is rank one,
the RIR can be computed exactly under certain conditions (Theorem~\ref{thm:dnrom}).

\appendix
\subsection{Proof of Lemma~\ref{lem:nec-rankdef}}
\label{prof:nec-rankdef}

Since $A$ is rank deficient with semisimple zero eigenvalue(s), 
there exists a nonsingular matrix $T$ such that 
\[
T^{-1}AT=\diag(\Lambda_k,0)
\]
for some $k\times k$ matrix $\Lambda_k$ with $k<n$. In light of \eqref{eq:P}, let 
\[
\hat P:=\hat T^{-1}P\hat T, \hs \hat T:=\diag(T,T), 
\]
and define $\hat\nabla$ and $\hat H$ similarly. Note that 
$\hat H=H$, and define the partitioned blocks of 
$\hat\nabla$ and $H$ as
\[
\hat\nabla=\diag(\hat\nabla_1,0), \hs
\hat\nabla_1:=\diag(\hat\Delta,\Lambda_k), \hs 
\hat\Delta:=T^{-1}\Delta T,
\] \[
H=\begin{mat}{ccc|c} w_{11}I & 0 & w_{12}I & 0 \\
0 & w_{11}I & 0 & w_{12}I \\
w_{21}I & 0 & hI & 0 \\ \hline
0 & w_{21}I & 0 & hI
\end{mat}=:
\begin{mat}{cc} H_{11} & H_{12} \\ H_{21} & H_{22} \end{mat},
\] 
where $\hat\nabla_1$ and $H_{11}$ are $(n+k)\times(n+k)$. Then
\[
\hat P = (I-H\hat\nabla)^{-1}H =
\begin{mat}{cc} \Phi H_{11} & \Phi H_{12} \\
H_{21}\Psi & H_{21}\Psi\hat\nabla_1H_{12}+H_{22}
\end{mat}
\] 
\[
\Phi:=(I-H_{11}\hat\nabla_1)^{-1}, \hs
\Psi:=(I-\hat\nabla_1H_{11})^{-1}
\]
By supposition, $\hat P$ and $\hat\nabla_1$ are stable.

Suppose $w_{12}(s)$ is unstable. Let $\alpha$ be a bounded signal
such that $H_{12}\alpha$ is unbounded. Define $\beta:=\Phi H_{12}\alpha$.
Then $\beta$ is bounded due to stability of $\hat P$, and satisfies
$\gamma=H_{12}\alpha$, where $\gamma:=(I-H_{11}\hat\nabla_1)\beta$.
Let $\gamma$ be partitioned as $\gamma=\col(*,\gamma_o,*)$ compatibly with $H_{11}$. Then $\gamma_o=w_{12}\alpha$.
The signal $w_{12}\alpha$ is unbounded by the choice of $\alpha$,
contradicting the fact that $\gamma$ is bounded because $\beta$ is bounded and the second row
block of $I-H_{11}\hat\nabla_1$ is stable due to $w_{11},\Delta\in\RHinf$.
Thus, it must hold that $w_{12}\in\RHinf$. 

Suppose $w_{21}(s)$ is unstable. Let $\gamma:=\col(0,\gamma_o,0)$ be a bounded 
signal such that $\beta:=H_{21}\gamma$ is unbounded. Define 
$\alpha:=(I-\hat\nabla_1H_{11})\gamma$. Then $\alpha$ is bounded because
the second column block of $I-\hat\nabla_1H_{11}$ is stable due to 
$w_{11},\Delta\in\RHinf$. By construction, bounded $\alpha$ and unbounded
$\beta$ satisfy $\beta=H_{21}\Psi\alpha$, which implies that $H_{21}\Psi$
is unstable. This contradicts stability of $\hat P$, and hence we conclude that
$w_{21}\in\RHinf$.

Finally, $h\in\RHinf$ follows by noting that $H_{21}\Psi$, $\hat\nabla_1$, 
$H_{12}$, and the $(2,2)$ block of $\hat P$ are all stable.

\subsection{Technical Results}

\begin{lemma} \label{lemma:IS2}
Suppose that the network system $\Sigma(\Delta, H, A)$ satisfies 
Assumptions~\ref{ass:stand} and \ref{ass:stable}. 
Then, $\Sigma(\Delta, H, A)$ is internally stable if and only if 
\begin{equation}
\label{eq:Hnabla}
(I - H(s)\nabla(s))^{-1}  \in \RHinf^{2n \times 2n} , 
\end{equation} 
or equivalently  
\begin{equation}
\label{eq:F}
F(s) :=  (M(s)-\Delta(s)N(s))^{-1} 
\in \RHinf^{n \times n},
\end{equation}
where $N(s)M(s)^{-1} = \LFT(H(s),A)$.
\end{lemma}

\begin{proof} 
Condition \eqref{eq:Hnabla} is a direct consequence of Theorem 5.5 in \cite{zhou1998}, 
since both $H$ and $\nabla$ are in $\RHinf$. 
Hence, we here show the equivalence of \eqref{eq:Hnabla} and \eqref{eq:F} to complete the proof. 
Let 
\begin{equation}
Q := I - H\nabla  = 
\begin{bmatrix} I - H_{11}\Delta  & -H_{12}A \\ -H_{21}\Delta & I - H_{22} A \end{bmatrix} , 
\end{equation} 
and note that $Q_{12} := -H_{12}A$ and $Q_{21} := H_{21}\Delta$ are both 
in $\RHinf$ and $Q_{11}^{-1} := (I - H_{11}\Delta )^{-1}$ is also in $\RHinf$
under Assumption~\ref{ass:stable}.
Applying an inverse formula %
\begin{equation}
\begin{bmatrix} Q_{11}  & Q_{12} \\ Q_{21} & Q_{22} \end{bmatrix}^{-1}  
=   
\begin{bmatrix}  \hat{Q}_{11} & -Q_{11}^{-1}Q_{12}\hat{Q}_{22} \\ 
-\hat{Q}_{22}Q_{21}Q_{11}^{-1}  &  \hat{Q}_{22} \end{bmatrix}
\end{equation} 
where 
\begin{eqnarray} 
\hat{Q}_{11} &=& %
 Q_{11}^{-1}  + Q_{11}^{-1} Q_{12}\hat{Q}_{22}Q_{21} Q_{11}^{-1},  \\
\hat{Q}_{22} &=& (Q_{22} - Q_{21}Q_{11}^{-1}Q_{12})^{-1},
\end{eqnarray} 
It is clear that the stability of $Q^{-1}$ is equivalent to 
\begin{eqnarray}
\hat{Q}_{22} &=& [(I - H_{22}A) - H_{21}\Delta(I - H_{11}\Delta)^{-1}H_{12}A]^{-1}  \nonumber \\ 
&=& [(I - h A) - w_{12}w_{21} (I - w_{11} \Delta)^{-1} \Delta A]^{-1}  \nonumber \\ 
&=& F (I - w_{11} \Delta)  \in \RHinf . 
\end{eqnarray}
Finally, the stability of $\hat{Q}_{22}$ is equivalent to that of $F$,  
since $I - w_{11} \Delta$ is stable and has no unstable zeros. 
\end{proof}

\begin{lemma} \label{lem:1}
Let a real matrix $A\in\IR^{n\times n}$ be given. Denote the rank of $A$ by $r$.
Then the following statements are equivalent:
\begin{itemize}
\item[(i)]  There exists a complex diagonal matrix $D$ such that $D^{-1}AD$ is normal.
\item[(ii)] There exist matrices $L,R^*\in\IC^{n\times r}$ and 
diagonal matrices $\Lambda\in\IC^{r\times r}$ and $W\in\IC^{n\times n}$ such that
\begin{equation} \label{A}
A=L\Lambda R, \hs RL=I, \hs \|RW\|\cdot\|W^{-1}L\|\leq1.
\end{equation}
\end{itemize}
\end{lemma}
\begin{proof}
Suppose (i) holds. Define $W:=D$ and a normal matrix $M:=W^{-1}AW$. 
Let $M=U\Lambda U^*$ be a spectral decomposition of $M$ such that 
$\Lambda\in\IC^{r\times r}$ is diagonal, and $U\in\IC^{n\times r}$ is a sub-block 
matrix of a unitary matrix. Then (ii) follows by noting that $\|U\|\leq1$ holds 
and setting $L:=WU$ and $R:=U^*W^{-1}$. Thus (i) $\Rightarrow$ (ii). 
To show the converse, suppose (ii) holds. Let $D:=W$ and $M:=W^{-1}AW$, and define 
$\hat L:=W^{-1}L$ and $\hat R:=RW$. Then, since $\hat R\hat L=I$ holds and 
$X=\hat R^{\dag}$ (the Moore-Penrose inverse) is the smallest spectral norm 
solution of $\hat RX=I$, we have
\[
1\geq\|\hat R\|\cdot\|\hat L\|\geq
\|\hat R\|\cdot\|\hat R^{\dag}\|=\sigma_{\max}/\sigma_{\min},
\]
where $\sigma_{\max}$ and $\sigma_{\min}$ are the max and min singular values of 
$\hat R$. This implies 
\[
\sigma_{\max}=\sigma_{\min} \hs \Rightarrow \hs I\leq \hat R\hat R^*\leq I
\hs \Rightarrow \hs \hat R\hat R^*=I.
\]
By a similar argument, we have $\hat L^* \hat L=I$. Then it follows that
$\hat R=\hat L^*$ holds since a direct calculation shows 
$(\hat L^*-\hat R)(\hat L^*-\hat R)^*=0$.
Finally $M=\hat L\Lambda\hat L^*$ and $\hat L^*\hat L=I$ imply that $M$ is normal. Thus we have (ii) $\Rightarrow$ (i).
\end{proof}

\begin{lemma} \label{lem:dnrom}
Let a real matrix $A\in\IR^{n\times n}$ be given. Suppose $A$ has rank one.
Then the following statements are equivalent:
\begin{itemize}
\item[(i)]  There exists a complex diagonal matrix $D$ such that $D^{-1}AD$ is normal.
\item[(ii)] All the nonzero diagonal entries of $A$ have the same sign.
\end{itemize}
\end{lemma}
\begin{proof}
Suppose (i) holds. Let a spectral decomposition of $M:=D^{-1}AD$ be given by
$M=\lambda uu^*$, where $\lambda\in\IR$ and $u\in\IC^n$. Define $\hat\ell:=Du$ and
$\hat r:=u^*D^{-1}$. Then $\hat\ell_i\hat r_i=|u_i|^2>0$ if $u_i\neq0$ and 
$\hat\ell_i=\hat r_i=0$ if $u_i=0$.
Since $A=\lambda\hat\ell\hat r$ is a real matrix, $\hat\ell_i\hat r_j$ is real for 
all $i$ and $j$. This implies that there exist $\ell,r\t\in\IR^n$ and $\phi\in\IR$ 
such that $\hat\ell=\ell e^{j\phi}$ and $\hat r=re^{-j\phi}$. Therefore, we have
$A=\lambda\ell r$ with either $\ell_ir_i>0$ or $\ell_i=r_i=0$ for all $i$. Thus, we 
conclude (i) $\Rightarrow$ (ii). 
To show the converse, suppose (ii) holds. Then matrix $A$ admits a factorization 
$A=\lambda\ell r$ with $\lambda\in\IR$ and $\ell,r\t\in\IR^n$ such that either
$\ell_ir_i>0$ or $\ell_i=r_i=0$ holds for all $i=1,\ldots,n$.
Let $d_i:=\sqrt{\ell_i/r_i}$ if $r_i\neq0$ and $d_i:=1$ if $r_i=0$. Define
\[
M:=D^{-1}AD=\lambda D^{-1}\ell rD=\lambda \hat\ell \hat r, \hs
\hat\ell_i:=\ell_i/d_i, \hs \hat r_i:=d_ir_i.
\]
Note that $\hat\ell_i=d_ir_i=\hat r_i$ if $r_i\neq0$, and $\hat\ell_i=\hat r_i=0$ 
if $r_i=0$. Hence we have $\hat\ell=\hat r\t$ and $M=\lambda \hat\ell\hat\ell\t$.
Since $M$ is symmetric, it is normal, and we have (ii) $\Rightarrow$ (i). 
\end{proof}

\end{document}